\documentclass[12pt,techreport,onecolumn]{IEEEtran}

\usepackage[final]{graphicx}
\usepackage[reqno]{amsmath}
\usepackage{amssymb}
\usepackage{url}
\usepackage{color}
\usepackage{times}
\usepackage{wrapfig}
\usepackage{subfig}
\newfont{\bbb}{msbm10 scaled 500}

\newfont{\bb}{msbm10 scaled 1100}

\newcommand{\EE}{\mbox{\bb E}}

\newcommand{\Prob}{\textrm{Pr}}

\newcommand{\Mc}{{\cal M}}

\newcommand{\Xc}{{\cal X}}
\newcommand{\Yc}{{\cal Y}}

\newtheorem{theorem}{Theorem}
\newtheorem{lemma}[theorem]{Lemma}%[chapter]
\title{Expansion Coding: Achieving the Capacity of an AEN Channel}

\author{O.~Ozan~Koyluoglu, Kumar~Appaiah, Hongbo~Si, 
and Sriram~Vishwanath
\thanks{The authors are with the
Laboratory for Informatics, Networks, and Communications (LINC),
Wireless Networking and Communications Group (WNCG),
The University of Texas at Austin,
1 Univerity Station, C0806, Austin, TX 78712.
Email: \{ozan,a.kumar,sihongbo\}@mail.utexas.edu, sriram@austin.utexas.edu.}
\thanks{A version of this report is submitted to
IEEE for a possible publication, after which this version 
may become obsolete.}
}

\begin{document}

\maketitle

%%%%%%%%%%%%%%%%%%%%%%%%%%%%%%%%%%%%%%%%%%%%%%%%%%%%%%%%%%%%%%%%%%%%%%%%%%%%%%
%%%%%%%%%%%%%%%%%%%%%%%%%%%%%%%%%%%%%%%%%%%%%%%%%%%%%%%%%%%%%%%%%%%%%%%%%%%%%%

\begin{abstract}
A general method of coding over expansions is proposed, which allows one
to reduce the highly non-trivial problem of coding over continuous channels
to a much simpler discrete ones. More specifically, the focus is on the
additive exponential noise (AEN) channel, for which the (binary) expansion of the
(exponential) noise random variable is considered. It is
shown that each of the random variables in the expansion corresponds to
independent Bernoulli random variables. Thus, each of the expansion levels 
(of the underlying channel) corresponds to a binary symmetric channel (BSC), 
and the coding problem is reduced to coding over these parallel channels 
while satisfying the channel input constraint. This optimization formulation 
is stated as the achievable rate result, for which a specific choice of input 
distribution is shown to achieve a rate which is  arbitrarily close to the 
channel capacity in the high SNR regime. Remarkably, the scheme allows
for low-complexity capacity-achieving codes for AEN channels, using
the codes that are originally designed for BSCs. Extensions to different
channel models and applications to other coding problems are discussed.
\end{abstract}

%%%%%%%%%%%%%%%%%%%%%%%%%%%%%%%%%%%%%%%%%%%%%%%%%%%%%%%%%%%%%%%%%%%%%%%%%%%%%%
%%%%%%%%%%%%%%%%%%%%%%%%%%%%%%%%%%%%%%%%%%%%%%%%%%%%%%%%%%%%%%%%%%%%%%%%%%%%%%

\section{Introduction}

In this work, we propose the method of constructing (binary) expansions of 
discrete-time continuous alphabet channels, and coding over the resulting 
set of parallel channels. We apply this \emph{coding over expansions} method 
to additive exponential noise (AEN) channels, where the signal and noise 
terms constructing the channel output are represented with their corresponding 
binary digits. Focusing on the additive exponential noise component, we 
show that the binary expansion of the noise consists of independent 
Bernoulli distributed random variables at each level. (The mean of each
random variable is a function of their level number in the expansion and 
the mean of the underlying exponential noise.) This way, thanks to the 
expansion technique, we show that continuous alphabet channels can 
be considered as a set of parallel binary symmetric channels (BSCs).

Instead of coding for every level in the expansion, we consider signaling
over a finite number of levels, and we resolve the problem arising from
carryovers either by considering them as noise, or by decoding them
the least significant bit onwards to the most
significant bit. For each case, we state the corresponding achievable
rate as an optimization problem, where the rate is maximized over
the choices of the Bernoulli distributions for the signal transmission over
each level with the constraint that that the combined random variables satisfy
the channel input constraint. Then, utilizing an approximately optimal input
distribution, we show that one can celebrate the achievability of an
$\epsilon$-gap to capacity result in the high SNR regime. This method
together with capacity-achieving low-complexity codes (such as polar
coding) allows one to achieve the capacity of the AEN in the high SNR regime.

The additive exponential noise (AEN) channel is of particular interest as it 
models worst-case noise given a mean and a non-negativity constraint on 
noise \cite{Verdu:Exponential96}. In addition, the AEN model naturally 
arises in non-coherent communication settings, and in optical communication 
scenarios. (We refer to \cite{Verdu:Exponential96} and 
\cite{Martinez:Communication11} for an extensive discussion on the 
AEN channel.) Verd{\'u} derived the optimal input distribution and the capacity
of the AEN channel in \cite{Verdu:Exponential96}. Martinez, on the other 
hand, proposed the pulse energy modulation scheme, which can be seen as a 
generalization of amplitude modulation for the Gaussian channels. 
In this scheme, the constellation symbols are chosen as $c (i-1)^l$, 
for $i=1, \cdots, 2^M$ with a constant $c$, and it is shown that 
the information rates obtained from this constellation can achieve an 
energy (SNR) loss of $0.76$ dB (with the best choice of 
$l=\frac{1}{2}(1+\sqrt{5})$) compared to the capacity in the high SNR 
regime. Another constellation technique for this coded modulation
approach is recently considered in \cite{LeGoff:Capacity11}, where
it is shown that log constellations are designed such that
the real line is divided into ($2M-1$) equally probable intervals.
$M$ of the centroids of these intervals are chosen as constellation
points, and, by a numerical computation of the mutual information, it
is shown that these constellations can achieve within a $0.12$ dB SNR gap
in the high SNR regime. Our approach, which achieves arbitrarily close to 
the capacity of the channel, outperforms these previously proposed
modulation techniques.

The rest of the paper is organized as follows. The next section describes the
AEN channel model. In Section~\ref{sec:BinExp}, we present the key
Lemma which shows that one independent Bernoulli distributed random
variables occur in the binary expansion of an exponential random variable.
Armed with this result, Section~\ref{sec:EM} develops the expansion
modulation technique, where we state the main results of the paper.
Numerical results are provided in Section~\ref{sec:NumRes}, and the paper
concludes with a discussion section (Section~\ref{sec:Discuss}).

%%%%%%%%%%%%%%%%%%%%%%%%%%%%%%%%%%%%%%%%%%%%%%%%%%%%%%%%%%%%%%%%%%%%%%%%%%%%%%
%%%%%%%%%%%%%%%%%%%%%%%%%%%%%%%%%%%%%%%%%%%%%%%%%%%%%%%%%%%%%%%%%%%%%%%%%%%%%%

\section{Channel model and background}

We consider the additive exponential noise (AEN) channel given by
\begin{equation}\label{eq:AENChannel}
Y_i=X_i+N_i, i=1,\cdots,n,
\end{equation}
where $N_i$ are independently and identically distributed
according to the additive noise with an exponential density with mean $E_N$;
i.e., omitting the index $i$, the noise has the following density:
\begin{equation}
\Prob\{N=n\} = \frac{1}{E_N} e^{-\frac{n}{E_N}} u(n),
\end{equation}
where $u(n)=1$ for $n\geq 0$ and $u(n)=0$ otherwise.

The transmitter conveys one of the messages, $m$, which is
uniformly distributed in $\Mc$, i.e., the random message
$m\in\Mc\triangleq \{1, \cdots, 2^{nR}\}$; and it does so by
mapping the message to the channel input using the encoding
function $f(\cdot):\Mc \to \Xc^n$, where $X_1^n(m)=f(m)$, under
the constraints that $\Xc=\Re$ and
\begin{equation}
\frac{1}{n} \EE\big[\sum\limits_{i=1}^n X_i \big] \leq E_X,
\end{equation}
where $E_X$ is the maximum average energy.

The decoder uses the decoding function $g(\cdot)$ to map its channel
observations to an estimate of the message. Specifically,
$g(\cdot):\Yc^n \to \Mc$, where the estimate is denoted
by $\hat{M} \triangleq g(Y^n)$.

The rate $R$ is said to be achievable, if the average
probability of error defined by
\begin{equation}
P_e \triangleq \frac{1}{|\Mc|} \sum\limits_{m\in \Mc}
\Prob\{g(Y^n)\neq m | m \textrm{ is sent. }\}
\end{equation}
can be made small for large $n$.
The capacity of AEN is denoted by $C$, which is
the maximum achievable rate $R$.

The capacity of AEN is given by \cite{Verdu:Exponential96}, where
\begin{equation}
C = \log(1+\textrm{SNR}),
\end{equation}
where $\textrm{SNR}=\frac{E_X}{E_N}$, and the capacity
achieving input distribution is given by
\begin{equation}\label{eq:OptInput}
p^*(x) = \left(\frac{E_X}{(E_X+E_N)^2} e^{\frac{-x}{E_X+E_N}}
+ \frac{E_N}{(E_X+E_N)} \delta(x)\right)u(x),
\end{equation}
where $\delta(x)=1$ if $x=0$, and $0$ otherwise.
Note that this is the $p^*(x)=\arg \max\limits_{p(x)} I(X;Y)$,
where $p(y|x)$ is given by the AEN channel \eqref{eq:AENChannel}.

Surprisingly, while the capacity achieving input distribution for
the additive white Gaussian noise (AWGN) channel is Gaussian, here
the optimal input distribution is not exponentially distributed.
However, we observe that in the high SNR regime, the optimal
distribution gets closer to an exponential distribution with mean
$E_X+E_N$.

%%%%%%%%%%%%%%%%%%%%%%%%%%%%%%%%%%%%%%%%%%%%%%%%%%%%%%%%%%%%%%%%%%%%%%%%%%%%%%
%%%%%%%%%%%%%%%%%%%%%%%%%%%%%%%%%%%%%%%%%%%%%%%%%%%%%%%%%%%%%%%%%%%%%%%%%%%%%%

\section{Exponential Distribution: Binary Expansion}\label{sec:BinExp}

We show the following lemma, which allows us to have independent
Bernoulli random variables in the binary expansion of an exponential
random variable.
\begin{lemma}\label{lem:exp}
Let $B_l$'s be independent Bernoulli random variables with parameters
given by $p_l$, i.e., $\Prob\{B_l=1\}=p_l$ and $\Prob\{B_l=0\}=1-p_l$,
and consider the random variable defined by
\begin{equation}
B=\sum\limits_{l=-\infty}^{\infty} 2^l B_l.
\end{equation}
Then, the choice of
\begin{equation}
p_l = \frac{1}{1+e^{\lambda 2^l}},
\end{equation}
implies that the random variable $B$ is exponentially distributed
with mean $\lambda^{-1}$, i.e.,
\begin{equation}
\Prob\{B=b\} = \lambda e^{-\lambda b} u(b).
\end{equation}
\end{lemma}
\begin{proof}
The proof follows by extending the one given in
\cite{Marsaglia:Random71}. We show the result by calculating the
moment generating function of $B$.
Using the assumption that $\{B_l\}_{l\in\mathbb{Z}}$ are
mutually independent, we have
\begin{eqnarray}
M_B(t)  =\EE[e^{tB}]
        =\EE\left[e^{t\sum\limits_{l=-\infty}^{\infty}2^l B_l}\right]
        =\prod_{l=-\infty}^{\infty}\EE\left[e^{t2^l B_l}\right].
\end{eqnarray}
Note that for any $l\in\mathbb{Z}$,
\begin{eqnarray}
\EE\left[e^{t2^l B_l}\right]&=&
\frac{e^{t2^l}}{1+e^{\lambda 2^l}}+\left(1-\frac{1}{1+e^{\lambda 2^l}}\right)
=\frac{e^{t 2^l}+e^{\lambda 2^l}}{1+e^{\lambda 2^l}}\nonumber\\
&=&\frac{1+e^{(t-\lambda ) 2^l}}{1+e^{-\lambda 2^l}}.
\end{eqnarray}
Then, using the fact that for any constant $\alpha\in\mathbb{R}$,
\begin{eqnarray}
\prod_{l=0}^{n}(1+e^{\alpha 2^l})=\frac{1-e^{2^{n+1}\alpha}}{1-e^{\alpha}},
\end{eqnarray}
we can obtain the following for $t<\lambda$,
\begin{eqnarray}
\prod_{l=0}^{\infty} \EE\left[e^{t2^l B_l}\right]
&=&\lim_{n\rightarrow\infty}\prod_{i=0}^{n}\frac{1+e^{(t-\lambda ) 2^l}}{1+e^{-\lambda 2^l}}\nonumber\\
&=&\frac{1-e^{-\lambda}}{1-e^{t-\lambda}}\nonumber\\
&=&\lim_{n\rightarrow\infty}\frac{1-e^{(t-\lambda)2^{n+1}}}{1-e^{t-\lambda}}\frac{1-e^{-\lambda}}
{1-e^{-\lambda2^{n+1}}}\nonumber\\
&=&\frac{1-e^{-\lambda}}{1-e^{t-\lambda}}.\label{eqn:part1}
\end{eqnarray}
And, similarly, for the negative part, we have
\begin{eqnarray}
\prod_{l=-n}^{-1}(1+e^{\alpha 2^l})=\frac{1-e^{\alpha}}{1-e^{\alpha2^{-n}}},
\end{eqnarray}
which further implies that
\begin{eqnarray}
\prod_{i=-\infty}^{-1}\EE\left[e^{t2^i B_l}\right]
&=&\lim_{n\rightarrow\infty}\frac{1-e^{t-\lambda}}{1-e^{(t-\lambda)2^{-n}}}\frac{1-e^{-\lambda2^{-n}}}{1-e^{-\lambda}}\nonumber\\
&=&\frac{\lambda(1-e^{t-\lambda})}{(\lambda-t)(1-e^{-\lambda})}.
\label{eqn:part2}
\end{eqnarray}
Thus, finally for any $t<\lambda$, multiplying equations~\ref{eqn:part1} and~\ref{eqn:part2}, we get
\begin{eqnarray}
M_B(t)&=&\frac{\lambda}{\lambda-t}.
\end{eqnarray}
The observation that this is the moment generation function 
for an exponentially distributed random variable with parameter 
$\lambda$ concludes the proof.
\end{proof}

%%%%%%%%%%%%%%%%%%%%%%%%%%%%%%%%%%%%%%%%%%%%%%%%%%%%%%%%%%%%%%%%%%%%%%%%%%%%%%
%%%%%%%%%%%%%%%%%%%%%%%%%%%%%%%%%%%%%%%%%%%%%%%%%%%%%%%%%%%%%%%%%%%%%%%%%%%%%%

\section{Modulation with expansions}\label{sec:EM}

Our proposed coding scheme consists of coding over expansion,
referred to as expansion modulation (EM), and coding for the resulting
parallel binary symmetric channels (BSCs). More specifically,
EM refers to a modulation over the binary expansion of the channel
over levels ranging from $-L_1$ to $L_2$, i.e.,
\begin{equation}\label{eq:EMChannel}
Y_i \triangleq \sum\limits_{l=-L_1}^{L_2} 2^{l} Y_{i,l}
= \sum\limits_{l=-L_1}^{L_2} 2^{l} (X_{i,l}+N_{i,l}),
\end{equation}
where the expansions of the signal and noise are given by
\begin{equation}
X_i \triangleq \sum\limits_{l=-L_1}^{L_2} 2^{l} X_{i,l}, \quad
\textrm{ and }
N_i \triangleq \sum\limits_{l=-L_1}^{L_2} 2^{l} N_{i,l}.
\end{equation}
When we take the limit $L_1,L_2\to\infty$, the channel given
by \eqref{eq:EMChannel} corresponds to the one given by
\eqref{eq:AENChannel}.

We propose coding over the expansion levels of this channel.
More specifically, the least significant bit channel is given by
\begin{equation}
Y_{i,-L_1} = X_{i,-L_1} \oplus N_{i,-L_1}, i=1,\cdots N.
\end{equation}
A capacity achieving BSC code is utilized over this channel with
input probability distribution given by $p_{-L_1}$.
(For example, the polar coding method~\cite{Arikan:Channel08},
allows one to construct capacity achieving codes for the $l=-L_1$
level channel.) Instead of directly using the capacity achieving
code design, we use the combination of the capacity achieving
code and the method of Gallager~\cite{Gallager:Information68} to
achieve a rate corresponding to the one obtained by the mutual
information $I(X_l;Y_l)$ evaluated with an input distribution
Bernoulli with parameter $p_l$.(The desired distributions
will be made clear in the following part.)

Noting that the sum is a modulo-$2$ sum in the above channel,
there will be carryovers from this sum to the next level,
$l=-L_1+1$. Denoting the carryover seen at level $l$ as
$C_{i,l}$, the remaining channels can be represented with
with the following
\begin{equation}\label{eq:ParCh}
Y_{i,l} = X_{i,l} \oplus \tilde{N}_{i,l}, \quad i=1,\cdots n,
\end{equation}
where the effective noise, $\tilde{N}_{i,l}$,
is a Bernoulli random variable obtained
by the convolution of the noise and the carryover
$\tilde{q}_{l} \triangleq \Prob\{ \tilde{N}_{i,l}=1\}
= q_{l} \otimes c_{l}
\triangleq q_l(1-c_l)+c_l(1-q_l)$ with
$q_{l} \triangleq \Prob\{ N_{i,l}=1\}$
and
$c_{l} \triangleq \Prob\{ C_{i,l}=1\}$.
Here, the carry over probability is given by
\begin{equation}
c_{l} = p_{l-1} \tilde{q}_{l-1}, l\in\{-L_1+1,\cdots, L_2\}.
\end{equation}

Due to Lemma~\ref{lem:exp}, the noise seen at each level
will be described by independent Bernoulli random variables, and therefore,
our coding scheme will be over the parallel channels given by
\eqref{eq:ParCh} for $l=-L_1,\cdots,L_2$, where
$\tilde{q}_{l} \triangleq \Prob\{ \tilde{N}_{i,l}=1\} =
q_{l} \otimes (p_{l-1} \tilde{q}_{l-1})$.

%%%%%%%%%%%%%%%%%%%%%%%%%%%%%%%%%%%%%%%%%%%%%%%%%%%%%%%%%%%%%%%%%%%%%%%%%%%%%%
%%%%%%%%%%%%%%%%%%%%%%%%%%%%%%%%%%%%%%%%%%%%%%%%%%%%%%%%%%%%%%%%%%%%%%%%%%%%%%

\subsection{Considering carryovers as noise}

Using a capacity achieving code for BSCs,
combined with the Gallager's method, expansion modulation
readily achieves the following result.
\begin{theorem}\label{thm:1}
Expansion modulation, when implemented with capacity achieving
codes for the resulting BSCs, achieves the rate given by
\begin{equation}
R_1=\sum\limits_{-L_1}^{L_2} H(p_l\otimes \tilde{q}_l) - H(\tilde{q}_l),
\end{equation}
for any $L_1,L_2>0$, where
$\tilde{q}_{l}=q_{l} \otimes (p_{l-1} \tilde{q}_{l-1})$
for $l>-L_1$. To satisfy the energy constraint, $p_l\in[0,0.5]$
is chosen such that
\begin{equation}
\frac{1}{n} \EE\left[\sum\limits_{i=1}^n X_i \right]
 = \frac{1}{n} \sum\limits_{i=1}^n \sum\limits_{l=-L_1}^{L_2} 2^{l} p_l
\leq E_X.
\end{equation}
\end{theorem}

The optimization problem stated in the result above is highly
non-trivial. However, utilizing the optimal input distribution,
(\ref{eq:OptInput}), one can adopt the following approximate distributions
in Theorem~\ref{thm:1}. At a high SNR, we observe that the optimal input
distribution is approximated by an exponential distribution. Then, one
can simply choose $p_l$ from the binary expansion of the exponential
distribution with mean $E_X+E_N$. To satisfy the power constraint, we
use coding only for $\frac{E_X}{E_X+E_N}$ of the time,
and for the rest we set the channel input to
$0$. (As a second approach, in the numerical results, we compare this
choice with the one of choosing $p_l$ from the binary expansion
of the exponential distribution with mean $E_X$.)
The next method gives a better rate result with a minimal increase
in complexity.

%%%%%%%%%%%%%%%%%%%%%%%%%%%%%%%%%%%%%%%%%%%%%%%%%%%%%%%%%%%%%%%%%%%%%%%%%%%%%%
%%%%%%%%%%%%%%%%%%%%%%%%%%%%%%%%%%%%%%%%%%%%%%%%%%%%%%%%%%%%%%%%%%%%%%%%%%%%%%

\subsection{Decoding carryovers}

In the scheme above, let us consider decoding starting from the level
$l=-L_1$. The receiver will obtain the correct $X_{i,-L_1}$ for
$i=1,\cdots,n$. As the receiver has the knowledge of $Y_{i,-L_1}$ for
$i=1,\cdots,n$, it will also have the knowledge of
the correct noise sequence $N_{i,-L_1}$ for $i=1,\cdots,N$.
With this knowledge, the receiver can directly obtain
$C_{i,-L_1+1}$ for $i=1,\cdots,N$, which is the carryover from
level $l=-L_1$ to level $l=-L_1+1$. Using this carryover sequence
in decoding at level $l=-L_1+1$, the receiver can get rid of
carryover noise. Thus, the effective channel that
the receiver will see can be represented by
\begin{equation}
Y_{i,l} = X_{i,l} \oplus N_{i,l}, \quad i=1, \cdots, n,
\end{equation}
for $l=-L_1,\cdots,L_2$.
We remark that with this decoding strategy the effective channels
will no longer be a set of independent parallel channels,
as decoding in one level affects the channels at higher levels.
However, if the utilized coding method is strong enough (e.g.,
if the error probability decays to $0$ exponentially with $n$),
then this carryover decoding error can be made insignificant
by increasing $n$ for a given number of levels (here, $L_1+L_2+1$).
We state the rate resulting from this approach.

\begin{theorem}\label{thm:2}
Expansion modulation, by decoding the carryovers, achieves
the rate given by
\begin{equation}
R_2=\sum\limits_{-L_1}^{L_2} H(p_l\otimes q_l) - H(q_l),
\end{equation}
for any $L_1,L_2>0$, where $p_l\in[0,0.5]$ is chosen to
satisfy
\begin{equation}
\frac{1}{n} \EE\big[\sum\limits_{i=1}^n X_i \big]
 = \frac{1}{n} \sum\limits_{i=1}^n \sum\limits_{l=-L_1}^{L_2} 2^{l} p_l
\leq E_X.
\end{equation}
\end{theorem}

Compared to the previous case, the optimization problem is simpler
here as the rate expression is simply the sum of the rates obtained
from a set of parallel channels.

We now show that the proposed scheme achieves the
capacity of AEN channel in the high SNR regime for a sufficiently
high number of levels. Towards this end, we provide a bound for
the capacity gap, $\Delta C\triangleq C-\hat{C}$, where
\begin{align}
\hat{C}=\frac{\textrm{SNR}}{1+\textrm{SNR}}\sum_{l=-L}^{L+\gamma}
\left[ H(p_l\otimes q_l)-H(q_l) \right]
\end{align}
is the achievable rate given in Theorem~\ref{thm:2}
with $L_1=L$, $L_2=\gamma+L$, and $\gamma=\log(1+\textrm{SNR})$,
when the approximate input distribution discussed above
(i.e., exponential with mean $E_X+E_N$) is used. First, we obtain the asymptotic behavior of entropy at each level.
\begin{lemma}\label{lem:entropy}
Entropy of the Bernoulli random variable at level $l$, $H(q_l)$,
is bounded by
\begin{align}
&H(q_l)< c_1\cdot2^{-l}\quad\;\;\;\;\text{for }\;l\geq 0,\label{lemma:entropy_0}\\
&H(q_l)> 1-c_2\cdot2^l\quad\text{for }\;l\leq 0,\label{lemma:entropy_1}
\end{align}
where $c_1$ and $c_2$ are both constants taking the values $c_1=3\log e$ and $c_2=\log e$.
\end{lemma}
\begin{proof}
Note that,
\begin{align}
H(q_l)  &=-q_l\log q_l-(1-q_l)\log (1-q_l)\nonumber\\
        &=-\frac{1}{1+e^{2^l}}\log \frac{1}{1+e^{2^l}}-\frac{e^{2^l}}{1+e^{2^l}}\log \frac{e^{2^l}}{1+e^{2^l}}.\nonumber
\end{align}
When $l\leq 0$, we obtain a lower bound as
\begin{align}
H(q_l)  &=\log \left(1+e^{2^l}\right) -\frac{e^{2^l}}{1+e^{2^l}}\log e \cdot2^l\nonumber\\
        &>\log(1+1)-\log e\cdot 2^l\nonumber\\
        &=1-\log e \cdot2^l.\nonumber
\end{align}
On the other hand, when $l\geq 0$, by using the facts that $\log (1+x)< \log e\cdot x$ for any $0<x<1$, and $e^x > 1+x+x^2/2$ for any $x>0$, we have
\begin{align}
H(q_l)  &=\frac{1}{1+e^{2^l}}\log \left(1+e^{2^l}\right)+\frac{e^{2^l}}{1+e^{2^l}}\log \left(1+e^{-2^l}\right)\nonumber\\
        &<\frac{1}{1+e^{2^l}}\log\left(2e^{2^l}\right)+ \log \left(1+e^{-2^l}\right)\nonumber\\
        &<\frac{1+2^l\log e}{1+e^{2^l}}+ e^{-2^l}\log e\nonumber\\
        &<\frac{1+2^l\log e}{1+1+2^l+2^{2l}/2}+ \frac{1}{1+2^l}\log e\nonumber\\
        &<2^{-l}2\log e + 2^{-l}\log e\nonumber\\
	&=2^{-l}3\log e .\nonumber
\end{align}
\end{proof}

%%%%%%%%%%
%FigSim 2%
%%%%%%%%%%
\begin{figure}[t]
 \centering
 \includegraphics[width=0.9\columnwidth]{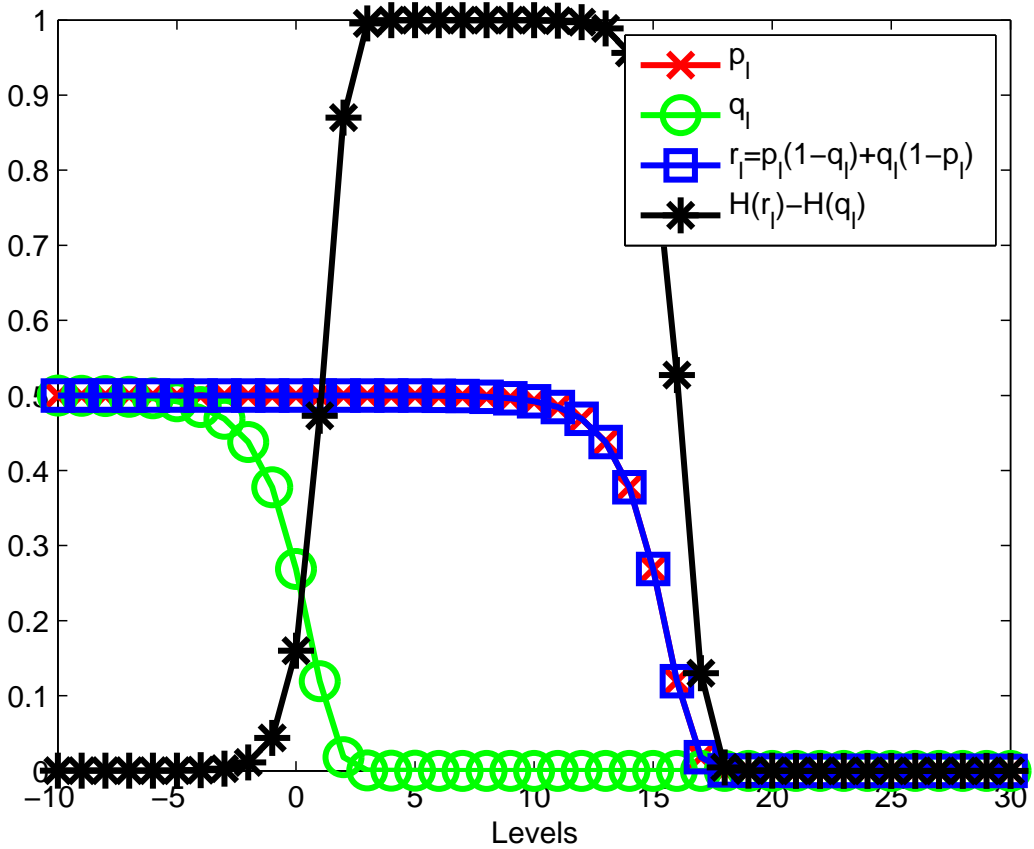}
 \caption{{\bf Signal and noise probabilities, and rate per level.}
$p_l$, $q_l$, $p_l\otimes q_l$ and rate at each level are shown,
where $\gamma=15$. The coding scheme with $L=5$ covers the significant portion of the rate obtained by
using all of the parallel channels. As shown in the text,
$p_l$ is a shifted version of $q_l$.
}
\label{fig:BitRate}
\end{figure}
%%%%%%%%%%

This lemma tells us that the tails bounds are exponential. Although better bounds may
exist, the exponential bound is sufficient for further analysis.
Based on the Lemma above, we obtain the following.
\begin{theorem}\label{thm:gap}
For any $\epsilon>0$, there exists an $\epsilon$-dependent constant
$c=\log\frac{16\log e}{\epsilon},$
such that if $\gamma\geq 2c$ and $L\geq c$, then the capacity
gap is bounded by $\Delta C\leq\epsilon$.
(The total number of levels is given by $2L+\gamma+1$, where
$\gamma=\log(1+\textrm{SNR})$.)
\end{theorem}
\begin{proof}
We first observe that
\begin{align}
\sum_{l=-L}^{L+\gamma}\left[ H(p_l\otimes q_l)-H(q_l)\right]
&=\sum_{l=-L}^{L+\gamma} \left[H(p_l\otimes q_l)-H(p_l)\right]
+\sum_{l=-L}^{L+\gamma} \left[H(p_l)-H(q_l)\right],\nonumber\\
&\geq\sum_{l=-L}^{L+\gamma} \left[H(p_l)-H(q_l)\right],\label{thm:C_r}
\end{align}
where the last inequality is due to the fact that $H(p_l\otimes q_l)\geq H(p_l)$.
Observing that
\begin{align}
p_{l+\gamma}=\frac{1}{1+e^{2^l}}=q_l,\label{thm:pk_nk}
\end{align}
and adopting Lemma~\ref{lem:entropy}, we have
\begin{enumerate}
\item[(1)] $c\leq l\leq \gamma-c$:
$H(p_l)-H(q_l)>1-c_2\cdot 2^{l-\gamma}-c_1\cdot 2^{-l},$

\item[(2)] $-c< l< c$:
$\left[H(p_l)-H(q_l)\right]+[H(p_{l+\gamma})-H(q_{l+\gamma})]=
H(p_l)-H(n_{l+\gamma})$
$>1-c_2\cdot 2^{l-\gamma}-c_1\cdot 2^{-(l+\gamma)},$ and

\item[(3)] $-L\leq l\leq-c$ and $\gamma+c\leq l\leq\gamma+L$:
$H(p_l)-H(q_l)>0.$
\end{enumerate}

Combining these pieces together, we will obtain
\begin{align}
\sum_{l=-L}^{L+\gamma} \left[H(p_l\otimes q_l)-H(p_l)\right]
&>\sum_{l=c}^{\gamma-c} \left(1-c_2\cdot 2^{l-\gamma}-c_1\cdot 2^{-l}\right)
+\sum_{l=-c+1}^{c-1}\left(1-c_2\cdot 2^{l-\gamma}-c_1\cdot 2^{-(l+\gamma)}\right)\nonumber\\
&>\gamma -(c_1+c_2)2^{-c}-(c_1+c_2)2^{c-\gamma-1}\nonumber\\
&>\gamma-8\log e\cdot 2^{-c},\label{thm:term2}
\end{align}
where the last inequality uses the assumption that $\gamma\geq 2c$. Thus, we obtain a bound for the capacity gap

\begin{align}
\Delta C&=\gamma -\frac{2^{\gamma}-1}{2^{\gamma}}\sum_{l=-L}^{L+\gamma} \Bigg\{\left[ H(p_l\otimes q_l)-H(p_l) \right]
+\left[ H(p_l)-H(q_l) \right]\Bigg\}\nonumber\\
        &\leq \gamma-\left(1-2^{-\gamma}\right)(\gamma-8\log e\cdot 2^{-c})\nonumber\\
        &<\gamma 2^{-\gamma}+8\log e\cdot 2^{-c}\nonumber\\
        &\leq2c\cdot 2^{-2c}+8\log e\cdot 2^{-c}\nonumber\\
        &<16\log e\cdot 2^{-c} \nonumber\\
        &=\epsilon,
\end{align}
where we used the decreasing property of
$f(x)=x\cdot 2^{-x}$ for $x>1$, and the assumption that
$\gamma\geq2c$. Fig.~\ref{fig:BitRate} helps explicate the
key steps of the proof.

\end{proof}

%%%%%%%%%%%%%%%%%%%%%%%%%%%%%%%%%%%%%%%%%%%%%%%%%%%%%%%%%%%%%%%%%%%%%%%%%%%%%%
%%%%%%%%%%%%%%%%%%%%%%%%%%%%%%%%%%%%%%%%%%%%%%%%%%%%%%%%%%%%%%%%%%%%%%%%%%%%%%

\subsection{Numerical results}\label{sec:NumRes}

We calculate the rates obtained from the two schemes above
($R_1$ in Theorem~\ref{thm:1} and $R_2$ in Theorem~\ref{thm:2})
with two different input probability distribution choices
(denoted by $C_1$ and $C_2$):
\begin{itemize}
\item $C_1$: Choosing $p_l$ from the binary expansion of
the exponential distribution with mean $E_X+E_N$.
To satisfy the power constraint, we
use coding only for the fraction of the channel uses
(i.e., $E_X/(E_X+E_N)$ of the time).
\item $C_2$: Choosing $p_l$ from the binary expansion of
the exponential distribution with mean $E_X$.
\end{itemize}

The first choice closely resembles the optimal distribution
given in \eqref{eq:OptInput}.
However, as the unused channels vanish in the high SNR
regime, we expect that both choices result in the same rate
as SNR gets large.
Numerical results are given in Fig.~\ref{fig:NumRes}.
It is evident from the figure (and from the analysis
given in Theorem~\ref{thm:gap}) that the proposed technique,
when implemented with sufficiently large number of levels,
outperforms the SNR gaps previously reported in
\cite{Martinez:Communication11} and
\cite{LeGoff:Capacity11}.

%%%%%%%%%%
%FigSim 1%
%%%%%%%%%%
\begin{figure}[t]
 \centering
 \includegraphics[width=0.9\columnwidth]{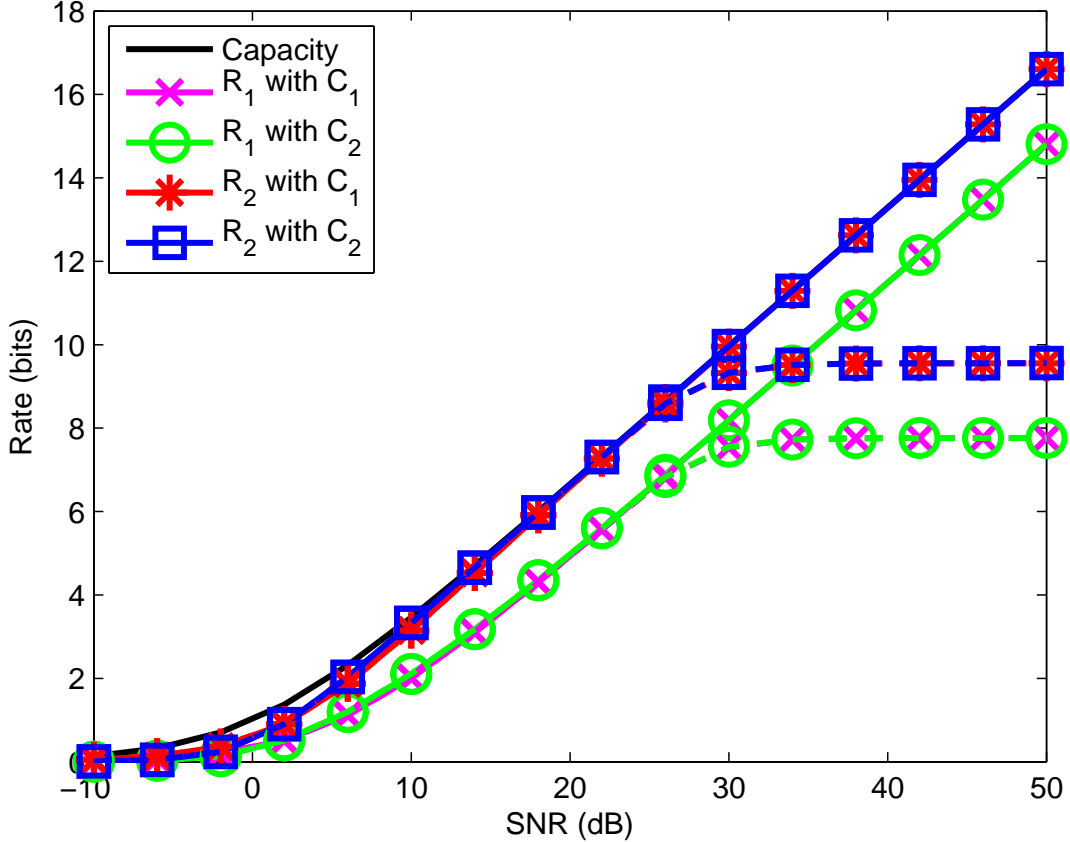}
 \caption{{\bf Numerical results.}
$R_1$: The rate obtained by considering carry over as noise.
$R_2$: The rate obtained by decoding carry overs at each level.
$C_1$: Choosing $p_l$ from the binary expansion of
the exponential distribution with mean $E_X+E_N$.
To satisfy the power constraint,
only a fraction of the channel uses
(i.e., $\frac{E_X}{E_X+E_N}$ of the time) is utilized.
$C_2$: Choosing $p_l$ from the binary expansion of
the exponential distribution with mean $E_X$.
Solid and dotted curves correspond to coding over
$41$ and $21$ number of levels, respectively.}
\label{fig:NumRes}
\end{figure}
%%%%%%%%%%

%%%%%%%%%%%%%%%%%%%%%%%%%%%%%%%%%%%%%%%%%%%%%%%%%%%%%%%%%%%%%%%%%%%%%%%%%%%%%%
%%%%%%%%%%%%%%%%%%%%%%%%%%%%%%%%%%%%%%%%%%%%%%%%%%%%%%%%%%%%%%%%%%%%%%%%%%%%%%

\section{Discussion}\label{sec:Discuss}

We note the followings.

\begin{itemize}
\item Expansion coding allows the construction of good channel codes for
discrete-time continuous channels using good discrete memoryless
channel codes. For instance, one can utilize binary expansion together
with polar codes. The underlying code is $q$-ary polar code, as we need to
implement Gallager's method in constructing the input distribution
$p(x)=\textrm{Ber}(p_l)$ for coding over level $l$.
(See, e.g., \cite{Arikan:Channel08,Sasoglu:Polarization09} for details.)
In addition, expansion modulation can be implemented over a
$q$-ary expansion of the channel, and any good code for the
resulting modulo $q$-sum channel can be used.

\item Avestimehr et al.~\cite{Avestimehr:Wireless11} have
introduced the deterministic approximation approach for
point-to-point and multi-user channels. The basic idea
is to construct an approximate channel for which the
transmitted signals are assumed to be noiseless above the noise level.
Using this high SNR approximation, one only needs to deal with the
interference seen at a particular receiver (in a networked model).
Our expansion coding scheme can be seen as a generalization
of these deterministic approaches. Here, the (effective) noise
in the channel is carefully calculated and the system takes
advantage of coding over the noisy levels at any SNR. This
generalized channel approximation approach can be useful in
reducing the large gaps reported in the previous works.

\item Although our discussion is limited to the AEN channel, where the
proposed scheme outperforms the previously proposed modulation schemes
and performs arbitrarily close to the capacity of the AEN channel, the
expansion coding refers to a more general framework and is not
limited to such channels. Towards this end, our ongoing efforts are focused
on utilizing the proposed scheme for AEN multiple-user channels,
and for their AWGN counterparts.

\end{itemize}

%%%%%%%%%%%%%%%%%%%%%%%%%%%%%%%%%%%%%%%%%%%%%%%%%%%%%%%%%%%%%%%%%%%%%%%%%%%%%%
%%%%%%%%%%%%%%%%%%%%%%%%%%%%%%%%%%%%%%%%%%%%%%%%%%%%%%%%%%%%%%%%%%%%%%%%%%%%%%

\appendices

%%%%%%%%%%%%%%%%%%%%%%%%%%%%%%%%%%%%%%%%%%%%%%%%%%%%%%%%%%%%%%%%%%%%%%%%%%%%%%
%%%%%%%%%%%%%%%%%%%%%%%%%%%%%%%%%%%%%%%%%%%%%%%%%%%%%%%%%%%%%%%%%%%%%%%%%%%%%%

\section{Proof of Lemma~\ref{lem:entropy}}
\label{sec:App-lem:entropy}
Recall that, in our construction, $q_l$ is the probability that Bernoulli random variable takes value 1 at level $l$. Thus, 
\begin{align}
H(q_l)  &=-q_l\log q_l-(1-q_l)\log (1-q_l)\nonumber\\
        &=-\frac{1}{1+e^{2^l}}\log \frac{1}{1+e^{2^l}}-\frac{e^{2^l}}{1+e^{2^l}}\log \frac{e^{2^l}}{1+e^{2^l}}.\nonumber
\end{align}
When $l\leq 0$, we obtain a lower bound as
\begin{align}
H(q_l)  &=\log \left(1+e^{2^l}\right) -\frac{e^{2^l}}{1+e^{2^l}}\log e \cdot2^l\nonumber\\
        &>\log(1+1)-\log e\cdot 2^l\nonumber\\
        &=1-\log e \cdot2^l.\nonumber
\end{align}
On the other hand, when $l\geq 0$, by using the facts that $\log (1+x)< \log e\cdot x$ for any $0<x<1$, and $e^x > 1+x+x^2/2$ for any $x>0$, we have
\begin{align}
H(q_l)  &=\frac{1}{1+e^{2^l}}\log \left(1+e^{2^l}\right)+\frac{e^{2^l}}{1+e^{2^l}}\log \left(1+e^{-2^l}\right)\nonumber\\
        &<\frac{1}{1+e^{2^l}}\log\left(2e^{2^l}\right)+ \log \left(1+e^{-2^l}\right)\nonumber\\
        &<\frac{1+2^l\log e}{1+e^{2^l}}+ e^{-2^l}\log e\nonumber\\
        &<\frac{1+2^l\log e}{1+1+2^l+2^{2l}/2}+ \frac{1}{1+2^l}\log e\nonumber\\
        &<2^{-l}2\log e + 2^{-l}\log e\nonumber\\
        &=2^{-l}3\log e .\nonumber
\end{align}
\noindent

%%%%%%%%%%%%%%%%%%%%%%%%%%%%%%%%%%%%%%%%%%%%%%%%%%%%%%%%%%%%%%%%%%%%%%%%%%%%%%
%%%%%%%%%%%%%%%%%%%%%%%%%%%%%%%%%%%%%%%%%%%%%%%%%%%%%%%%%%%%%%%%%%%%%%%%%%%%%%

%%%%%%%%%%%%%%%%%%%%%%%%%%%%%%%%%%%%%%%%%%%%%%%%%%%%%%%%%%%%%%%%%%%%%%%%%%%%%%
%%%%%%%%%%%%%%%%%%%%%%%%%%%%%%%%%%%%%%%%%%%%%%%%%%%%%%%%%%%%%%%%%%%%%%%%%%%%%%

\end{document}